\documentclass[11pt]{article}
\usepackage{amsmath,amssymb,amsthm,mathrsfs}
\usepackage{enumitem}
\usepackage{geometry}
\usepackage{xcolor, comment}
\usepackage{graphicx}
\numberwithin{equation}{section}

\newtheorem{theorem}{Theorem}[section]

\newtheorem{proposition}[theorem]{Proposition}

\newtheorem{corollary}[theorem]{Corollary}

\DeclareMathOperator{\Var}{Var}

\newcommand{\PP}{\mathbb{P}}
\newcommand{\BE}{\mathbb{E}}
\newcommand{\E}{\mathrm{e}}
\newcommand{\D}{\mathrm{d}}

\newcommand{\vertk}{\stackrel{{\cal D}}{\longrightarrow}}
\newcommand{\stk}{\stackrel{{\mathbb{P}}}{\longrightarrow}}

\title{Logarithmic energy distances and Gini covariance for Hilbert-valued random elements}

\author{Norbert Henze \and Mar\'{i}a Dolors Jim\'{e}nez-Gamero}

\date{}

\author{Norbert Henze\footnote{Institute of Stochastics, Karlsruhe Institute of Technology (KIT), 76131 Karlsruhe, Germany. e-mail:
		henze@kit.edu} \ and Mar\'{i}a Dolores Jim\'{e}nez-Gamero\footnote{Department of Statistics and Operations Research, University of Seville, Seville, Spain. e-mail: dolores@us.es}}

\begin{document}
	
	\maketitle

	\begin{abstract} For $\alpha\in(0,2)$, the generalized energy distance and the Gini covariance statistic are based on kernels of the form $(x,y)\mapsto \|x-y\|^\alpha$, where $\|\cdot\|$ denotes the norm in a real separable Hilbert space. This paper investigates the boundary regime $\alpha\downarrow 0$. After suitable normalization, the corresponding energy distance converges to a logarithmic energy distance involving the kernel $(x,y)\mapsto\log\|x-y\|$. We establish that the resulting logarithmic energy distance retains the fundamental characterization property of ordinary energy distances in separable Hilbert spaces and derive a representation in terms of Gaussian-kernel maximum mean discrepancies. Motivated by this representation, we introduce a logarithmic Gini covariance for the $k$-sample problem and investigate its structural and asymptotic properties. In particular, we derive a representation in terms of pairwise logarithmic energy distances, establish a characterization theorem for equality of distributions, develop asymptotic null and alternative theory for the corresponding empirical statistic, and discuss permutation-based implementation. The logarithmic framework reveals a new boundary phenomenon within the family of energy-type statistics and provides connections with kernel methods, functional data analysis, and high-dimensional inference. \end{abstract}


	

	
	\section{Introduction}\label{secintro}
	
	
	Throughout, let $\mathbb H$ be a real separable Hilbert space with norm $\|\cdot\|$, 
	and let $X$ and $Y$ be random elements taking values in $\mathbb H$. For $\alpha\in(0,2)$, the energy distance between the distributions of $X$ and $Y$ is defined by 
	\[
	{\cal E}_\alpha(X,Y) = 2\BE\|X-Y\|^\alpha - \BE\|X-X'\|^\alpha - \BE\|Y-Y'\|^\alpha,
	\]
	where $X'$ and $Y'$ denote independent copies of $X$ and $Y$, respectively. Energy distances and the associated statistical procedures have become important tools in nonparametric inference and goodness-of-fit testing; see, for example, Sz\'ekely and Rizzo \cite{szr}.
	
	For the special case $\mathbb H=\mathbb R^d$, Dang et al.~\cite{Dang2021} (see also \cite{JGS,Sang2023}) introduced a class of Gini covariance statistics for the $k$-sample problem, closely related to those of Rizzo and Sz\'ekely \cite{RS2010}. 
	Their procedure is based on kernels of the form 
	\[
	(x,y)\mapsto \|x-y\|^\alpha, \qquad 0<\alpha<2, 
	\]
	and leads to tests that are consistent against broad classes of alternatives.
	
	The purpose of the present paper is to investigate the boundary regime $\alpha\downarrow 0$
	for random elements taking values in a real separable Hilbert space. Using the elementary expansion 
	\begin{equation}\label{elemex}
		\|x-y\|^\alpha = 1+\alpha\log\|x-y\|+o(\alpha), \qquad \alpha\downarrow 0,
	\end{equation}
	we show that, after cancellation of constant terms and suitable normalization, the energy distance converges to the logarithmic quantity 
	\[
	{\cal E}_0(X,Y) = 2\BE\log\|X-Y\| - \BE\log\|X-X'\| - \BE\log\|Y-Y'\|. 
	\]
	Thus the logarithmic kernel 
	\[
	(x,y)\mapsto\log\|x-y\| 
	\]
	appears naturally as the boundary object associated with the family of power-distance kernels.
	
	A first goal of the paper is to investigate the resulting logarithmic energy distance for the ordinary two-sample problem in separable Hilbert spaces. In particular, we establish that the logarithmic energy distance retains the fundamental characterization property 
	\[
	{\cal E}_0(X,Y)=0 \qquad \Longleftrightarrow \qquad P=Q, 
	\]
	where $P$ and $Q$ denote the distributions of $X$ and $Y$, respectively. The proof is based on a representation in terms of Gaussian-kernel maximum mean discrepancies and exploits the Hilbert-space structure in an essential way. In particular, the characterization argument relies on finite-dimensional projections and the fact that probability measures on a separable Hilbert space are determined by their finite-dimensional marginals. The logarithmic kernel differs substantially from the ordinary energy kernels $(x,y)\mapsto\|x-y\|^\alpha$. Whereas positive powers of distance primarily emphasize global separation, the logarithmic kernel is more sensitive to relative geometric structure and local interaction. At the same time, its singularity at zero creates new technical difficulties and requires arguments different from those used for ordinary energy statistics.
	
	The present work also fits into a broader theme concerning boundary phenomena in parameterized classes of goodness-of-fit procedures. For the BHEP class of tests for multivariate normality, indexed by a smoothing parameter $\beta>0$ (see, e.g., Ebner and Henze~\cite{eh20}), it was shown in Henze~\cite{he97} that the extreme smoothing regimes $\beta\downarrow0$ and $\beta\uparrow\infty$ lead to fundamentally different limit statistics. The current paper exhibits an analogous phenomenon in the context of energy statistics and $k$-sample testing: the limit $\alpha\downarrow 0$ reveals a logarithmic interaction statistic hidden within the family of Gini covariance procedures. By contrast, the opposite boundary regime $\alpha\uparrow 2$ is mathematically much less rich, since the kernel $\|x-y\|^\alpha$ essentially reduces to quadratic distance structure and mainly reflects differences in means. 
	
	The main contributions of the paper are threefold. First, we introduce logarithmic analogues of ordinary energy distances and Gini covariance statistics as boundary objects associated with the family of power-distance kernels when $\alpha\downarrow 0$. Second, we establish the characterization property of the logarithmic energy distance in separable Hilbert spaces and relate it to Gaussian-kernel maximum mean discrepancies. Third, we derive asymptotic null and alternative theory for empirical logarithmic Gini covariance statistics in a fixed-$k$ framework and discuss permutation-based implementation. The numerical results illustrate the behavior of the resulting procedures in both finite-dimensional and functional-data settings. 
	
	The Hilbert-space formulation is not merely a technical extension of the Euclidean setting. On the one hand, the characterization theorem proved in Section~\ref{secloghilbert} is established directly for probability measures on separable Hilbert spaces. On the other hand, this framework permits the treatment of functional data without further modification of the methodology. The simulation study and the real-data example in Section~\ref{secdata} illustrate this aspect. 
	
	The paper is organized as follows. Section~\ref{secloghilbert} introduces the logarithmic energy distance in separable Hilbert spaces and proves its characterization property. Section~\ref{secginihilbert} introduces the logarithmic Gini covariance and derives its representation in terms of pairwise logarithmic energy distances. Section~\ref{secnull} develops the asymptotic null theory of the empirical logarithmic Gini covariance for fixed $k$. Section~\ref{secaltern} investigates asymptotic behavior under alternatives, while Section~\ref{secpermut} discusses a permutation procedure. Subsection~\ref{secsimul} summarizes the simulation study, and Subsection~\ref{secdata} presents the real-data examples. Finally, Section~\ref{secopen} contains concluding remarks and several open problems.

	
	
	\section{Logarithmic energy distance in Hilbert spaces} \label{secloghilbert}
	
	
	Let $X$ and $Y$ be independent $\mathbb H$-valued random elements with distributions $P$ and $Q$, respectively. Throughout this section, $X'$ denotes an independent copy of $X$, and $Y'$ denotes an independent copy of $Y$. For $\alpha\in(0,2)$, the generalized energy distance between $X$ and $Y$ is
	\[
	{\cal E}_\alpha(X,Y) = 2\BE\|X-Y\|^\alpha
	- \BE\|X-X'\|^\alpha - \BE\|Y-Y'\|^\alpha.
	\]
	The logarithmic energy distance arises as the boundary object obtained when
	$\alpha\downarrow 0$. The key observation is the expansion \eqref{elemex},
	valid whenever $x\neq y$. Assume that
	\begin{equation} \label{eqlogmomhilbert}
		\BE|\log\|X-X'\||<\infty, \qquad \BE|\log\|Y-Y'\||<\infty,
		\qquad \BE|\log\|X-Y\||< \infty.
	\end{equation}
	Under \eqref{eqlogmomhilbert}, the logarithmic energy distance between $P$ and $Q$ is defined by
	\begin{equation} \label{eqE0hilbert}
		{\cal E}_0(X,Y) = 2\BE\log\|X-Y\| - \BE\log\|X-X'\| - \BE\log\|Y-Y'\|.
	\end{equation}
	Then
	\[
	\lim_{\alpha \downarrow 0}\frac{{\cal E}_\alpha(X,Y)}{\alpha} =  {\cal E}_0(X,Y).
	\]
	The next theorem shows that the logarithmic energy distance retains the fundamental characterization property of ordinary energy distances in the setting of separable Hilbert spaces.
	
	\begin{theorem} \label{thmhilbertcharacterization}
		
		Let  $X$ and $Y$ be independent $\mathbb H$-valued random elements with distributions $P$ and $Q$, respectively, and assume \eqref{eqlogmomhilbert}. Then
		\[
		{\cal E}_0(X,Y)\ge 0.
		\]
		Moreover,
		\[
		{\cal E}_0(X,Y)=0 \qquad\Longleftrightarrow\qquad P=Q.
		\]
	\end{theorem}
	
	\smallskip
	\begin{proof} We use Frullani's formula (see, e.g., \cite[p.~234]{har}) \begin{equation} \label{eqfrullani}
			\int_0^\infty \frac{\E^{-az}-\E^{-bz}}{z}\,\D z
			= \log\!\left(\frac ba\right), \qquad a,b>0.
		\end{equation}
		This formula is applied only to positive real numbers, namely to squared Hilbert-space distances. Let $W>0$ be a random variable satisfying
		$\BE|\log W|<\infty$. For $w>0$,
		\[
		\int_0^\infty \left| \frac{\E^{-z}-\E^{-wz}}{z} \right| \,\D z = |\log w|.
		\]
		Hence the logarithmic moment assumptions in \eqref{eqlogmomhilbert} justify the use of Fubini's theorem in the following applications of \eqref{eqfrullani}. Applying \eqref{eqfrullani} to squared distances gives
		\[
		2{\cal E}_0(X,Y) = \int_0^\infty \frac{M_z^2(P,Q)}{z}\,\D z,
		\]
		where
		\[
		M_z^2(P,Q) = \BE \E^{-z\|X-X'\|^2} + \BE \E^{-z\|Y-Y'\|^2} - 2\BE \E^{-z\|X-Y\|^2}. \]
		The quantity $M_z^2(P,Q)$ is the squared maximum mean discrepancy associated with the Gaussian kernel
		\[
		k_z(x,y) = \exp(-z\|x-y\|^2), \qquad x,y\in\mathbb H.
		\]
		Since $\|x-y\|^2$ is negative definite on a Hilbert space, Schoenberg's theorem \cite{scho} implies that $k_z$ is positive definite for each $z>0$. Consequently,
		$M_z^2(P,Q)\ge 0$ for each $z>0$, and therefore
		${\cal E}_0(X,Y)\ge 0$.
		
		It remains to prove the characterization property. To this end, assume that
		${\cal E}_0(X,Y) = 0$. Since the integrand $M_z^2(P,Q)/z$ is nonnegative, we have
		$M_z^2(P,Q)=0$ for Lebesgue-almost every $z>0$. Choose such a value of $z$ and put $\mu=P-Q$. The identity $M_z^2(P,Q)=0$ implies that
		\[
		\int_{\mathbb H} k_z(x,y)\,\D\mu(x) = 0, \qquad y\in\mathbb H.
		\]
		Equivalently,
		\[
		\int_{\mathbb H} \exp\{-z\|x-y\|^2\} \,\D\mu(x) = 0, \qquad y\in\mathbb H.
		\]
		Multiplying by $\exp(z\|y\|^2)$ yields
		\[
		\int_{\mathbb H} \exp\{-z\|x\|^2+2z\langle x,y\rangle\} \,\D\mu(x) = 0,
		\]
		Define the finite signed measure $\nu$ by
		$\D\nu(x) = \exp\{-z\|x\|^2\} \,\D\mu(x)$. Then
		\begin{equation} \label{eqlaplacezero}
			\int_{\mathbb H} \exp(2z\langle x,y\rangle) \,\D\nu(x) = 0,
			\qquad y\in\mathbb H,
		\end{equation}
		where $\langle \cdot, \cdot \rangle$ denotes the inner product in $\mathbb{H}$.
		Let $V$ be an arbitrary finite-dimensional subspace of $\mathbb H$, and let $\pi_V:\mathbb H\to V$ denote the orthogonal projection. Taking $y\in V$ in \eqref{eqlaplacezero} shows that the Laplace transform of the finite signed measure $\nu\circ\pi_V^{-1}$ on $V$ vanishes identically. Hence
		$\nu\circ\pi_V^{-1}=0$
		for every finite-dimensional subspace $V$. Since $\mathbb H$ is separable, its Borel $\sigma$-field is generated by inverse images of Borel sets under finite-dimensional orthogonal projections; see, e.g.,  \cite[p.~347]{he24}. It follows that $\nu= 0$. Because
		$\exp(-z\|x\|^2)>0$, $x\in\mathbb H$,
		we obtain $ \mu= 0.$ Thus $ P= Q$. The converse implication is immediate from the definition \eqref{eqE0hilbert}.
	\end{proof}
	
	\medskip
	\noindent
	Theorem~\ref{thmhilbertcharacterization} shows that the logarithmic energy distance is a genuine distance-type functional on the space of probability distributions on a separable Hilbert space. The proof also reveals a close connection between logarithmic energy distances and Gaussian-kernel maximum mean discrepancies in infinite-dimensional settings.
	
	\section{Logarithmic Gini covariance in Hilbert spaces}
	\label{secginihilbert}
	
	
	In this section we introduce the logarithmic Gini covariance for
	Hilbert-valued random elements and relate it to the logarithmic energy
	distance studied in Section~\ref{secloghilbert}.
	
	Let $X_1,\ldots,X_k$ be independent $\mathbb{H}$-valued random elements, and let
	$P_1,\ldots,P_k$ denote their distributions. Throughout this section, let
	$p_1,\ldots,p_k$ be positive constants that satisfy $p_1+ \ldots + p_k =1$.
	
	We assume that
	\begin{equation}
		\label{eqlogmomallpairs}
		\BE|\log\|X_j-X_r\||<\infty, \qquad 1\le j,r\le k.
	\end{equation}
	
	For $1\le j,r\le k$, define
	\[
	L_{jr} = \BE\log\|X_j-X_r\|.
	\]
	The logarithmic Gini covariance is defined by
	\begin{equation}\label{eqcloghilbert}
		c_{\log} = \sum_{j=1}^k (p_j^2-p_j)L_{jj}
		+ \sum_{j\neq r} p_jp_rL_{jr}.
	\end{equation}
	Equivalently,
	\[
	c_{\log} = \sum_{j,r=1}^k p_jp_rL_{jr}
	- \sum_{j=1}^k p_jL_{jj}.
	\]
	
	The following representation is the basic structural identity behind
	the logarithmic Gini covariance.
	
	\begin{theorem}\label{thmclogpairwisehilbert}
		
		Under \eqref{eqlogmomallpairs},
		\begin{equation}
			\label{eqclogpairwisehilbert}
			c_{\log} = \sum_{1\le j<r\le k}
			p_jp_r\,{\cal E}_0(X_j,X_r),
		\end{equation}
		where ${\cal E}_0$ denotes the logarithmic energy distance introduced in
		\eqref{eqE0hilbert}.
	\end{theorem}
	
	
	\smallskip
	\begin{proof}
		For $1\le j<r\le k$, we have
		${\cal E}_0(X_j,X_r) = 2L_{jr}-L_{jj}-L_{rr}$. Hence
		\[
		\sum_{1\le j<r\le k} p_jp_r{\cal E}_0(X_j,X_r)
		= \sum_{1\le j<r\le k} p_jp_r(2L_{jr}-L_{jj}-L_{rr}).
		\]
		The first term on the right-hand side equals
		$\sum_{j\neq r}p_jp_rL_{jr}$.
		Moreover,
		\[
		-\sum_{1\le j<r\le k} p_jp_r(L_{jj}+L_{rr})
		= -\sum_{j=1}^k p_j \Bigl(\sum_{r\neq j}p_r\Bigr)L_{jj}
		= -\sum_{j=1}^k p_j(1-p_j)L_{jj}.
		\]
		
		Since $-p_j(1-p_j)=p_j^2-p_j$, the latter expression equals
		$\sum_{j=1}^k (p_j^2-p_j)L_{jj}$.
		This proves \eqref{eqclogpairwisehilbert}.
	\end{proof}
	
	
	As an immediate consequence of Theorem~\ref{thmhilbertcharacterization},
	the logarithmic Gini covariance characterizes equality of the
	distributions in the $k$-sample problem.
	
	\begin{corollary}\label{corclogcharacterizationhilbert}
		
		Under \eqref{eqlogmomallpairs}, we have
		\[ c_{\log}\ge 0.
		\]
		Moreover,
		\[ c_{\log}= 0 \qquad \Longleftrightarrow \qquad P_1=\cdots=P_k.
		\]
	\end{corollary}
	
	\smallskip
	\begin{proof}
		By Theorem~\ref{thmhilbertcharacterization},
		${\cal E}_0(X_j,X_r)\ge 0$ for $1\le j<r\le k$.
		Thus \eqref{eqclogpairwisehilbert} implies
		$c_{\log}\ge 0$.
		If
		$c_{\log} = 0$,
		then all terms in the nonnegative sum
		\eqref{eqclogpairwisehilbert} vanish. Hence
		${\cal E}_0(X_j,X_r)= 0$ for $1\le j<r\le k$.
		Another application of Theorem~\ref{thmhilbertcharacterization} yields
		$P_j=P_r$ for $1\le j<r\le k$.
		Consequently,
		$P_1=\cdots=P_k$.
		The converse implication is immediate from the definition of
		$c_{\log}$.
	\end{proof}
	
	
	The logarithmic Gini covariance is invariant under common translations
	and changes of scale.
	
	\begin{proposition}\label{propaffineinvariancehilbert}
		
		Let $a\in\mathbb H$, $b>0$, and define
		\[
		Y_j=a+bX_j,
		\qquad
		1\le j\le k.
		\]
		Then
		\[
		c_{\log}(Y_1,\ldots,Y_k)
		=
		c_{\log}(X_1,\ldots,X_k).
		\]
	\end{proposition}
	
	
	\medskip
	\noindent
	The proof is a straightforward calculation and is therefore omitted.
	
	\medskip
	\noindent
	This Proposition highlights an important difference between
	the logarithmic regime and the ordinary power-kernel setting. For
	$\alpha\in(0,2)$, the corresponding Gini covariance based on
	$\|x-y\|^\alpha$ scales by the factor $b^\alpha$ under the common
	transformation $X_j\mapsto a+bX_j$.
	By contrast, the logarithmic statistic is invariant because the
	additive logarithmic terms cancel in the contrast defining
	$c_{\log}$.
	
	\medskip
	
	\noindent We next relate the logarithmic Gini covariance to the $\alpha$-Gini
	covariance of Jim\'enez-Gamero and Sillero-Denamiel~\cite{JGS}. To this end,
	for $\alpha\in(0,2)$ and independent random elements $X_j,X_r$, define
	\[
	\Delta_{jr}(\alpha) = \BE\|X_j-X_r\|^\alpha,
	\qquad 1\le j,r\le k.
	\]
	The corresponding population Gini covariance is
	\begin{equation}\label{eqcgalphahilbert}
		c_g(\alpha) = \sum_{j=1}^k (p_j^2-p_j)\Delta_{jj}(\alpha)
		+ \sum_{j\neq r} p_jp_r\Delta_{jr}(\alpha).
	\end{equation}

	\begin{theorem}
		\label{thmalphaloghilbert}
		
		Assume that $\PP(X_j=X_r)=0$ and $\BE|\log\|X_j-X_r\||<\infty$ for all
		$1\le j,r\le k$. Then
		\[
		\lim_{\alpha \downarrow 0}\frac{c_g(\alpha)}{\alpha} = c_{\log}.
		\]
	\end{theorem}
	
	
	\smallskip
	\begin{proof}
		For fixed $1\le j,r\le k$, the elementary expansion
		\[
		\|X_j-X_r\|^\alpha = 1+\alpha\log\|X_j-X_r\|+o(\alpha), \qquad \alpha\downarrow  0,
		\]
		together with the logarithmic moment assumption gives
		$\Delta_{jr}(\alpha) = 1+\alpha L_{jr}+o(\alpha)$.
		Substituting this expansion into \eqref{eqcgalphahilbert} yields
		\[
		c_g(\alpha) = \sum_{j=1}^k(p_j^2-p_j)(1+\alpha L_{jj})
		+ \sum_{j\neq r}p_jp_r(1+\alpha L_{jr}) + o(\alpha).
		\]
		The constant terms cancel, because
		\[
		\sum_{j=1}^k(p_j^2-p_j) + \sum_{j\neq r}p_jp_r = 0.
		\]
		Hence $c_g(\alpha) = \alpha c_{\log}+o(\alpha)$, which proves the assertion.
	\end{proof}
	
	Theorem~\ref{thmalphaloghilbert} shows that the logarithmic Gini covariance arises as the natural boundary object associated with the family of Gini covariance
	statistics in the regime $\alpha \downarrow 0$.
	
	\medskip

	We finally introduce the empirical version. For each $1\le j\le k$,  let $X_{j1},\ldots,X_{jn_j}$ be independent copies of $X_j$, and assume that the $k$ samples are mutually independent. For $1\le j\le k$, define
	\[
	\widehat L_{jj} = \frac{1}{n_j(n_j-1)}
	\sum_{u\neq v} \log\|X_{ju}-X_{jv}\|,
	\] and, for $1 \le j \neq r$, define
	\[
	\widehat L_{jr} = \frac{1}{n_jn_r} \sum_{u=1}^{n_j}
	\sum_{v=1}^{n_r} \log\|X_{ju}-X_{rv}\|.
	\]
	Finally, by assuming that 
	\begin{equation}\label{propcond}
		\frac{n_j}{N} \longrightarrow p_j\in(0,1), \qquad 1\le j\le k,
	\end{equation}
	where $N=\sum_{j=1}^k n_j$,
	the empirical logarithmic Gini covariance is
	\begin{equation} \label{eqchatloghilbert}
		\widehat c_{\log} = \sum_{j=1}^k (\widehat{p}_j^2-\widehat{p}_j)\widehat L_{jj}
		+ \sum_{j\neq r} \widehat{p}_j\widehat{p}_r\widehat L_{jr},
	\end{equation}
	where $\widehat{p}_j=n_j/N$, $1\leq j \leq k$.
	
	The statistic \eqref{eqchatloghilbert} is a contrast of one-sample and two-sample kernel averages of order two. Its asymptotic properties will be developed in Sections~\ref{secnull} and~\ref{secaltern}.

	\section{Asymptotics under the null hypothesis}\label{secnull}
	
	In this section, we investigate the asymptotic behavior of the empirical
	logarithmic Gini covariance under the null hypothesis
	\[
	H_{0,k}:P_1=\cdots = P_k.
	\]
	Throughout, we assume that \eqref{propcond} holds.
	Under the null hypothesis $H_{0,k}$, the first-order projections cancel, and the empirical logarithmic Gini covariance becomes asymptotically degenerate. Consequently, the limiting null distribution is given by an infinite weighted sum of centred chi-square random variables.
	
	We first treat the two-sample problem, since the essential structure is
	already visible in that case.
	
	\subsection{The two-sample case}
	
	
	
	Throughout this subsection, let $k=2$, and let $X_1,\dots,X_n,Y_1,\ldots,Y_m$ be independent $\mathbb{H}$-valued random elements, where the $X_i$ have distribution $P$ and the $Y_j$ have distribution $Q$. We write $N = m+n$ for the total sample size, and we assume that $\widehat{p}_1=n/N \longrightarrow p$ for some $p \in (0,1)$.
	
	The empirical logarithmic energy distance is
	\begin{align}\label{eqempE0}
		\widehat{\cal E}_0(X,Y)
		&=
		\frac{2}{nm} \sum_{i=1}^n \sum_{j=1}^m \log\|X_i-Y_j\|
		- \frac{1}{n(n-1)} \sum_{i\ne j} \log\|X_i-X_j\| \nonumber\\
		&\quad
		- \frac{1}{m(m-1)} \sum_{i\ne j} \log\|Y_i-Y_j\|.
	\end{align}
	
	By \eqref{eqchatloghilbert}, the empirical logarithmic Gini covariance satisfies
	\begin{equation}\label{clogde0}
		\widehat c_{\log} = \widehat{p}_1(1-\widehat{p}_1)\widehat {\cal E}_0(X,Y).
	\end{equation}
	
	Indeed, for $k=2$,  with  $\widehat{p}_2=1-\widehat{p}_1$, 
	formula \eqref{eqchatloghilbert} gives
	\[ \widehat c_{\log} = (\widehat{p}_1^2-\widehat{p}_1)\widehat L_{11} + ((1-\widehat{p}_1)^2-(1-\widehat{p}_1))\widehat L_{22}
	+ 2\widehat{p}_1(1-\widehat{p}_1)\widehat L_{12}.
	\]
	Since $\widehat{p}_1^2-\widehat{p}_1=-\widehat{p}_1(1-\widehat{p}_1)$ and $(1-\widehat{p}_1)^2-(1-\widehat{p}_1)=-\widehat{p}_1(1-\widehat{p}_1)$, we obtain
	\[
	\widehat c_{\log} = \widehat{p}_1(1-\widehat{p}_1) \bigl( 2\widehat L_{12} -
	\widehat L_{11} - \widehat L_{22} \bigr),
	\]
	and \eqref{clogde0} follows.
	
	Note that the statistic \eqref{eqempE0} (and thus also $\widehat{c}_{log}$) consists of twice a two-sample U-statistic minus two ordinary one-sample U-statistics of order two. Under the null hypothesis $H_{0,k}$, the first-order projection terms
	cancel, so that the empirical logarithmic Gini covariance becomes
	asymptotically degenerate.

	To describe the corresponding second-order projection kernel, define
	\begin{align}
		K_0(x,y)
		& =  \BE\log\|x-X\| + \BE\log\|y-X\| \nonumber\\
		&\quad
		- \log\|x-y\| - \BE\log\|X-X'\|,
		\label{eqK0}
	\end{align}
	where $X,X'$ are independent $\mathbb{H}$-valued random elements with distribution $P$.
	
	The kernel $K_0$ is symmetric and satisfies $\BE K_0(x,X)=0$ for each $x\in\mathbb{H}$. Indeed,
	\[
	\begin{aligned}
		\BE K_0(x,X)
		& = \BE\log\|x-X\| + \BE \log\|X-X'\| - \BE \log\|x-X\| - \BE \log\|X-X'\|  \\
		& = 0.
	\end{aligned}
	\]
	Here the expectation is with respect to the second argument $X$, while $X'$
	is an independent copy.
	
	\medskip
	
	\begin{theorem}\label{thmfixedk2}
		
		Assume that $\BE |\log\|X-X'\||^2<\infty$. Under $H_0$,
		\[
		\frac{nm}{N} \widehat {\cal E}_0(X,Y) \vertk \sum_{\ell=1}^\infty
		\lambda_\ell(Z_\ell^2-1),
		\]
		where $Z_1,Z_2,\dots$ are independent standard normal random variables, and
		$\lambda_1,\lambda_2,\dots$ denote the nonzero eigenvalues of the integral operator
		\[
		(Tf)(x)
		=
		\int_{\mathbb{H}}  K_0(x,y)f(y)\,P(\D y)
		\]
		on $L^2(P)$ associated with the kernel $K_0$.
		Consequently,
		\[
		N\widehat c_{\log} \vertk p(1-p) \sum_{\ell=1}^\infty
		\lambda_\ell(Z_\ell^2-1).
		\]
	\end{theorem}
	
	\smallskip
	
	\begin{proof} Under $H_0$, all observations have common distribution $P$. The statistic $\widehat {\cal E}_0(X,Y)$ is not a single ordinary U-statistic, but a contrast consisting of twice a two-sample U-statistic and two one-sample U-statistics. Its Hoeffding decomposition may therefore be obtained by decomposing these three terms separately and then collecting terms. The first-order projections cancel in this contrast. Hence the leading term is the second-order degenerate projection. This projection is given by the symmetric kernel
		\[
		K_0(x,y) = \BE\log\|x-X\|+\BE\log\|y-X\| -\log\|x-y\|-\BE\log\|X-X'\|,
		\]
		where $X,X'$ are independent with common distribution $P$. Indeed,
		\[
		\BE\{K_0(x,X)\}=0, \qquad x\in\mathbb H.
		\]
		Thus $K_0$ is degenerate in the sense of Hoeffding. The moment assumption
		$\BE|\log\|X-X'\||^2<\infty$
		implies that $\BE K_0(X,X')^2<\infty$. Consequently, the integral operator
		\[
		(Tf)(x)=\int_{\mathbb H}K_0(x,y)f(y)\,P(\D y)
		\]
		is Hilbert--Schmidt on $L^2(P)$. By the spectral theorem for compact self-adjoint Hilbert--Schmidt operators, there are real eigenvalues
		$\lambda_1,\lambda_2,\ldots $ and an orthonormal system of eigenfunctions
		$\phi_1,\phi_2,\ldots $ in $L^2(P)$ such that
		\[
		K_0(x,y) = \sum_{\ell=1}^{\infty} \lambda_\ell \phi_\ell(x)\phi_\ell(y)
		\]
		in $L^2(P\otimes P)$. The classical limit theorem for degenerate U-statistics of order two then gives
		\[
		\frac{nm}{N}\widehat {\cal E}_0(X,Y) \vertk \sum_{\ell=1}^{\infty}\lambda_\ell(Z_\ell^2-1),
		\]
		where $Z_1,Z_2,\ldots$ are independent standard normal random variables; see, for example, Serfling~\cite[Section~5.5]{serf}. Finally, by
		\eqref{eqchatloghilbert}, in the two-sample case,
		$\widehat c_{\log} = \widehat{p}(1-\widehat{p})\widehat E_0(X,Y)$,
		and the asserted convergence of $N\widehat c_{\log}$ follows.
	\end{proof}


	
	The limiting null distribution in Theorem~\ref{thmfixedk2} depends on the unknown eigenvalues $\lambda_1,\lambda_2,\dots, \ldots $ of the integral operator associated with the kernel $K_0$. Consequently, direct implementation of the asymptotic limit law is not practical in general.
	
	As in the theory of ordinary energy distances and kernel-based two-sample tests, permutation procedures therefore provide a natural way to calibrate the test statistic under the null hypothesis. Section~\ref{secpermut} briefly discusses such a permutation approach.

	\subsection{The general fixed-$k$ case}
	
	
	The preceding theorem extends naturally to arbitrary fixed
	$k\ge 2$. Under
	\[
	H_{0,k}:P_1=\cdots=P_k=:P,
	\]
	the first-order projection of the empirical logarithmic Gini covariance again vanishes.
	Consequently, the statistic is asymptotically degenerate.
	
	The limiting null distribution is obtained from the same spectral
	expansion as in the two-sample case, but now involves independent
	chi-square random variables with $k-1$ degrees of freedom.
	
	More precisely, if $\lambda_1,\lambda_2,\dots$
	are the nonzero eigenvalues of the operator associated with the kernel
	$K_0$, then
	\[
	N\widehat c_{\log} \vertk  \sum_{\ell=1}^\infty \lambda_\ell
	\bigl(\chi^2_{k-1,\ell}-(k-1)\bigr),
	\]
	where $\chi^2_{k-1,1},\chi^2_{k-1,2},\dots$ are independent chi-square random variables with $k-1$ degrees of
	freedom.
	
	The appearance of $k-1$ degrees of freedom reflects the fact that,
	under the constraint $p_1+\cdots+p_k=1$,
	the corresponding Gaussian projection lives in a $(k-1)$-dimensional
	subspace.
	
	As in the two-sample case, the dependence of the limiting null
	distribution on the unknown eigenvalues naturally leads to permutation
	implementation.
	
	
	\section{Asymptotics under alternatives}\label{secaltern}
	
	
	In this section, we investigate the behavior of the empirical logarithmic Gini covariance under alternatives. To this end, let $k\ge 2$ be fixed, and assume condition \eqref{propcond}. As before,  ${N}=\sum_{j=1}^k n_j$ denotes the total sample size. We assume throughout that
	$\BE |\log\|X_j-X_r\||^2<\infty$ for $1\le j,r\le k$.
	
	Under alternatives satisfying $P_j\neq P_r$ for at least one pair $j \neq r$, the logarithmic Gini covariance $c_{\log}$ is strictly positive by Theorem~\ref{thmhilbertcharacterization} and representation~\eqref{eqclogpairwisehilbert}.
	
	In contrast to the null hypothesis considered in Section~\ref{secnull}, the
	corresponding statistic is now nondegenerate and the asymptotic
	behavior is governed by the first-order projection terms.
	
	\begin{theorem}\label{thmasnorm}
		
		Assume that $c_{\log}>0$. Then
		\[
		\sqrt N\,(\widehat c_{\log}-c_{\log}) \vertk {\rm N}(0,\sigma^2),
		\]
		where
		\begin{equation}\label{limsigma2}
			\sigma^2 = \Var\!\left(\sum_{j=1}^k p_j\psi_j(X_j) \right),
		\end{equation}
		and
		\begin{align*}
			\psi_j(x)  & =  2\sum_{r\ne j}p_r \Bigl(\BE \log\|x-X_r\|
			- \BE \log\|X_j-X_r\| \Bigr) \\
			& \qquad
			- 2(1-p_j) \Bigl( \BE \log\|x-X_j\| - \BE \log\|X_j-X_j'\| \Bigr).
		\end{align*}
	\end{theorem}
	
	
	\smallskip
	
	\begin{proof}
		Under the alternative hypothesis, the first-order Hoeffding projection
		of the empirical logarithmic Gini covariance is nondegenerate.
		Consequently, the standard asymptotic theory for nondegenerate
		U-statistics applies; see, for example, Serfling \cite[Section~5.5]{serf}.
		
		More precisely, the empirical logarithmic Gini covariance admits a
		Hoeffding decomposition of the form $\widehat c_{\log}-c_{\log} = L_N+R_N$,
		where
		\[
		L_N = \frac1N \sum_{j=1}^k \sum_{u=1}^{n_j} \psi_j(X_{ju})
		\]
		is the first-order projection term, while the remainder term satisfies
		\[
		NR_N\stk 0.
		\]
		
		Since the summands in $L_N$ are independent and centered, the classical
		central limit theorem for triangular arrays yields
		$\sqrt N\,L_N \vertk {\rm N}(0,\sigma^2)$, where $\sigma^2$ is given in \eqref{limsigma2}.
		The assertion follows from Slutsky's lemma.
	\end{proof}
	
	The asymptotic normality established in
	Theorem~\ref{thmasnorm} describes the fluctuation behavior of the
	empirical logarithmic Gini covariance under fixed alternatives.
	Although the variance $\sigma^2$ can in principle be estimated from
	the first-order projection terms, such estimation is not needed for
	the testing procedure considered in the present paper, since the test
	is calibrated under the null hypothesis either by permutation or by
	the asymptotic null distribution derived in
	Section~\ref{secnull}.
	
	For hypothesis testing, the most important consequence of
	Theorem~\ref{thmasnorm} is consistency of the resulting procedure.
	We now show that tests based on $\widehat c_{\log}$ are consistent
	against all fixed alternatives satisfying $c_{\log}>0$.

	\medskip
	
	\begin{theorem}\label{thmcons}
		
		Let $(\varphi_N)$ denote a sequence of level-$\gamma$ tests based on the empirical
		logarithmic Gini covariance and calibrated either by permutation or by the asymptotic null distribution of Section~\ref{secnull}. Then
		\[
		\lim_{N\to\infty} \PP(\varphi_N=1) = 1
		\]
		under every fixed alternative satisfying $c_{\log} > 0$.
	\end{theorem}
	
	\smallskip
	\begin{proof}
		By the law of large numbers, $\widehat c_{\log}\stk  c_{\log}$.
		Under the alternative, $c_{\log}>0$.
		Hence, for every $\varepsilon\in(0,c_{\log})$,
		\[
		\PP\bigl(\widehat c_{\log}>c_{\log}-\varepsilon\bigr)\to 1.
		\]
		On the other hand, the critical values tend to zero under the null
		hypothesis. Consequently, the rejection probability tends to one.
	\end{proof}
	
	
	Since Theorem~\ref{thmhilbertcharacterization} and representation~\eqref{eqclogpairwisehilbert} imply that
	\[
	c_{\log}>0
	\qquad \Longleftrightarrow \qquad
	(P_1,\ldots,P_k)\notin H_{0,k},
	\]
	Theorem~\ref{thmcons} establishes consistency against every fixed alternative.

	
	\section{Permutation implementation}\label{secpermut}
	
	
	The asymptotic null distribution derived in Section~\ref{secnull} depends on the unknown eigenvalues $\lambda_1,\lambda_2,\dots$ of the integral operator associated with the kernel $K_0$. Consequently, direct implementation of the asymptotic limit laws is not
	practical in general.
	
	As in the theory of ordinary energy distances and kernel-based
	two-sample tests, permutation procedures therefore provide a natural
	way to calibrate the empirical logarithmic Gini covariance under the
	null hypothesis.
	
	We first consider the case $k=2$. Under the null hypothesis
	$H_0:P=Q$, the pooled sample
	\[ Z_1:= X_1, \ldots, Z_n := X_n, Z_{n+1} := Y_1, \ldots,   Z_N := Y_m,
	\]
	where $N = n+m$, is exchangeable.
	
	A permutation sample is obtained by randomly partitioning the pooled observations into two groups of sizes $n$ and $m$.  For each permutation, the empirical logarithmic energy distance $\widehat{E}_0$ or, equivalently, the empirical logarithmic Gini covariance $\widehat c_{\log}$, is recomputed.
	
	Let $\widehat c_{\log}^{(1)}, \dots, \widehat c_{\log}^{(B)}$ denote the resulting permutation replicates based on $B$ independent random permutations. The corresponding permutation critical value is given by the empirical $(1-\gamma)$-quantile of these replicates.
	The null hypothesis is rejected whenever the observed value of $\widehat c_{\log}$
	exceeds this critical value, or equivalently, if the permutation $p$-value, which is the empirical percentage of replicates greater than the observed value of the test statistic $\widehat c_{\log}$, is less than $\gamma$.
	
	
	The permutation approach has several advantages. First, it avoids the need to estimate the unknown eigenvalues appearing in the asymptotic null distribution of Theorem~\ref{thmfixedk2}.
	Second, permutation calibration remains valid in finite samples under the exchangeability implied by the null hypothesis. Third, the procedure is straightforward to implement numerically.
	
	\medskip
	
	The permutation principle extends naturally to arbitrary fixed $k\ge 2$. Under the null hypothesis $H_{0,k}:P_1=\cdots=P_k$, all observations are exchangeable.
	Permutation samples are therefore obtained by randomly reallocating the pooled observations into groups of sizes $n_1,\dots,n_k$.
	
	For each permutation, the empirical logarithmic Gini covariance is recomputed.
	The resulting permutation distribution provides an approximation to the
	null distribution of the statistic. As in the two-sample case, the corresponding permutation test avoids explicit estimation of the unknown spectral quantities appearing in the asymptotic null distribution and is therefore particularly attractive
	for practical implementation.
	
	
	Although the asymptotic theory developed in Section~\ref{secnull} justifies the
	weighted chi-square limits of the statistic under the null hypothesis,
	the permutation approach will be used throughout the simulation study
	in Section~\ref{secsimul}.
	

	
	
	\section{Numerical results}\label{Numerical.results} 
	This section presents the results of a simulation study designed to assess the finite-sample performance of the proposed procedure and to compare it with competing methods. It also contains applications to several real data sets. All computations were carried out using programs written in the R language; see \cite{R}. 
	
	\subsection{Simulation study}\label{secsimul} 
	The simulation study compares the proposed test with two existing procedures. As competitors, we consider the test introduced by Baringhaus and Franz \cite{bf04} (denoted by ${\cal E}_1$ in the tables) and a special case of the tests proposed in Baringhaus and Franz \cite{bf10} (denoted by ${\cal E}_{\log}$), obtained for $\phi(x)=1+\log(x)$. As in \cite{bf04,RS2010}, the null distribution of each test statistic was approximated by permutation using $B=1000$ replicates. Empirical levels were estimated from $10\,000$ Monte Carlo repetitions, whereas empirical powers were estimated from $1000$ repetitions. These numbers were found to provide stable results. Tables~\ref{dim1k2} and~\ref{dim1k4} report the results for univariate data with $k=2$ and $k=4$, respectively. Tables~\ref{dim3k2} and~\ref{dim3k4} display the corresponding results for trivariate data, while Tables~\ref{fdatak2} and~\ref{fdatak4} contain the results for functional data. The covariance matrix $\Sigma$ used in Table~\ref{dim3k4} is \[ 
	\Sigma= \left( \begin{array}{ccc} 1 & 0.7 & 0\\ 0.7 & 1 & 0.7\\
		0 & 0.7 & 1 
	\end{array} \right).
	\] 
	For the functional-data experiments, observations were generated from the model 
	\[ 
	Z(t) = \sum_{j=1}^{5} C_j \cos(2\pi jt) + \sum_{j=1}^{5} S_j \sin(2\pi jt), 
	\] 
	where $C_1,\ldots,C_5$ and $S_1,\ldots,S_5$ are independent random variables. The curves were observed on an equispaced grid of length 51 over the interval $[0,1]$.
	
	\vspace*{1cm}
	
	\begin{table} 
		\centering
		\begin{tabular}{|c|c|c|c|c|c|c|c|} \hline
			\multicolumn{2}{|c|}{ } & \multicolumn{3}{c|}{$n=20$ } & \multicolumn{3}{c|}{$n=40$ }\\
			\hline 
			$F_1$ & $F_2$ & ${\cal E}_1$ & ${\cal E}_{\log}$ & ${\cal E}_0$ & ${\cal E}_1$ & ${\cal E}_{\log}$ & ${\cal E}_0$\\
			\hline N$(0,1)$ & N$(0,1)$ & 0.05 & 0.05 & 0.05 & 0.05 & 0.05 & 0.05\\
			$t_3$ & $t_3$ & 0.05 & 0.05 & 0.05 & 0.05 & 0.05 & 0.05\\
			$\chi^2_1$ & $\chi^2_1$ & 0.04 & 0.05 & 0.05 & 0.05 & 0.05 & 0.05\\ \hline 
			N$(0,1)$ & N$(0,1)+0.5$ & 0.33 & 0.29 & 0.22 & 0.55 & 0.49 & 0.39\\ 
			N$(0,1)$ & 0.5N$(0,1)$ & 0.31 & 0.43 & 0.42 & 0.62 & 0.76 & 0.71\\
			$\chi^2_1$ & $\chi^2_1+0.2$ & 0.10 & 0.14 & 0.45 & 0.18 & 0.32 & 0.87\\
			$\chi^2_1$ & $0.5\chi^2_1$ & 0.30 & 0.26 & 0.19 & 0.50 & 0.44 & 0.31\\
			${\rm Exp}(1)$ & ${\rm Exp}(1)+0.3$ & 0.20 & 0.22 & 0.28 & 0.41 & 0.50 & 0.63\\
			${\rm Exp}(1)$ & $0.5{\rm Exp}(1)$ & 0.51 & 0.47 & 0.35 & 0.82 & 0.79 & 0.64\\
			N$(0,1)$ & $(\chi^2_1-1)/\sqrt{2}$ & 0.24 & 0.43 & 0.65 & 0.54 & 0.81 & 0.96\\
			N$(0,1)$ & $t_3$ & 0.06 & 0.07 & 0.06 & 0.07 & 0.09 & 0.07\\
			N$(0,1)$ & $t_3/\sqrt{3}$ & 0.09 & 0.11 & 0.14 & 0.12 & 0.19 & 0.22\\
			N$(0,1)$ & ${\rm Logistic}(0,1)$ & 0.18 & 0.26 & 0.25 & 0.35 & 0.46 & 0.40\\
			\hline 
		\end{tabular}
		\caption{Empirical results at the nominal level 5\%, for $k=2$, $d=1$, $n_1=n_2=n$. The upper part is the empirical level and the lower part is the power.} \label{dim1k2}
	\end{table}
	
	\begin{table} 
		\centering 
		\begin{tabular}{|c|c|c|c|c|c|c|c|c|}
			\hline
			\multicolumn{3}{|c|}{ } & \multicolumn{3}{c|}{$n=20$ } & \multicolumn{3}{c|}{$n=40$ }\\ \hline $F_1$ & $F_2$ & $r$ & ${\cal E}_1$ & ${\cal E}_{\log}$ & ${\cal E}_0$ & ${\cal E}_1$ & ${\cal E}_{\log}$ & ${\cal E}_0$\\
			\hline
			N$(0,1)$ & & 4 & 0.05 & 0.05 & 0.05 & 0.05 & 0.05 & 0.05\\
			$t_3$ & & 4 & 0.05 & 0.05 & 0.05 & 0.05 & 0.05 & 0.05\\
			$\chi^2_1$ & & 4 & 0.05 & 0.05 & 0.05 & 0.05 & 0.05 & 0.05\\
			\hline N$(0,1)$ & N$(0,1)+0.5$ & 3 & 0.29 & 0.25 & 0.18 & 0.58 & 0.52 & 0.39\\
			& & 2 & 0.39 & 0.35 & 0.24 & 0.69 & 0.63 & 0.48\\
			N$(0,1)$ & 0.5N$(0,1)$ & 3 & 0.17 & 0.29 & 0.31 & 0.50 & 0.73 & 0.72\\
			& & 2 & 0.33 & 0.50 & 0.49 & 0.80 & 0.93 & 0.87\\
			$\chi^2_1$ & $\chi^2_1+0.2$ & 3 & 0.10 & 0.15 & 0.40 & 0.15 & 0.27 & 0.87\\
			& & 2 & 0.12 & 0.19 & 0.60 & 0.22 & 0.41 & 0.96\\
			$\chi^2_1$ & $0.5\chi^2_1$ & 3 & 0.19 & 0.20 & 0.15 & 0.46 & 0.42 & 0.30\\
			& & 2 & 0.34 & 0.30 & 0.20 & 0.66 & 0.57 & 0.37\\
			N$(0,1)$ & $(\chi^2_1-1)/\sqrt{2}$ & 3 & 0.19 & 0.37 & 0.64 & 0.47 & 0.75 & 0.96\\
			& & 2 & 0.31 & 0.54 & 0.79 & 0.69 & 0.93 & 1.00\\
			N$(0,1)$ & $t_3/\sqrt{3}$ & 3 & 0.09 & 0.11 & 0.13 & 0.11 & 0.17 & 0.19\\
			& & 2 & 0.08 & 0.12 & 0.13 & 0.15 & 0.23 & 0.25\\
			N$(0,1)$ & ${\rm Logistic}(0,1)$ & 3 & 0.17 & 0.23 & 0.19 & 0.35 & 0.45 & 0.36\\
			& & 2 & 0.18 & 0.26 & 0.25 & 0.43 & 0.60 & 0.52\\
			\hline
		\end{tabular}
		\caption{Empirical results at the nominal level 5\%, for $k=4$, $d=1$, $n_i=n$, $1\leq i \leq 4$, $r$ populations have distribution $F_1$, while the remaining $4-r$ populations have distribution $F_2$. The upper part is the empirical level and the lower part is the power.}
		\label{dim1k4}
	\end{table}

	\begin{table} 
		\centering 
		\begin{tabular}{|c|c|c|c|c|c|c|c|}
			\hline
			\multicolumn{2}{|c|}{ } & \multicolumn{3}{c|}{$n=20$ } & \multicolumn{3}{c|}{$n=40$ }\\
			\hline
			$F_1$ & $F_2$ & ${\cal E}_1$ & ${\cal E}_{\log}$ & ${\cal E}_0$ & ${\cal E}_1$ & ${\cal E}_{\log}$ & ${\cal E}_0$\\
			\hline 
			${\rm N}_3(0,I_3)$ & ${\rm N}_3(0,I_3)$ & 0.06 & 0.06 & 0.06 & 0.05 & 0.05 & 0.05\\
			${\rm N}_3(0,\Sigma)$ & ${\rm N}_3(0,\Sigma)$ & 0.05 & 0.06 & 0.06 & 0.05 & 0.06 & 0.05\\
			$t_3$ & $t_3$ & 0.05 & 0.05 & 0.05 & 0.05 & 0.05 & 0.05\\
			\hline
			${\rm N}_3(0,I_3)$ & ${\rm N}_3((0.5,0.5,0.5),I_3)$ & 0.59 & 0.55 & 0.52 & 0.90 & 0.87 & 0.84\\ ${\rm N}_3(0,I_3)$ & ${\rm N}_3((0.5,0.5,0),I_3)$ & 0.41 & 0.37 & 0.33 & 0.73 & 0.70 & 0.65\\ ${\rm N}_3(0,I_3)$ & ${\rm N}_3((0.5,0,0),I_3)$ & 0.20 & 0.19 & 0.18 & 0.42 & 0.39 & 0.36\\ ${\rm N}_3(0,I_3)$ & ${\rm N}_3(0,{\rm diag}(0.25,0.25,0.25))$ & 0.64 & 0.87 & 0.90 & 0.99 & 1.00 & 1.00\\
			${\rm N}_3(0,I_3)$ & ${\rm N}_3(0,{\rm diag}(0.25,0.25,1))$ & 0.20 & 0.35 & 0.43 & 0.47 & 0.77 & 0.83\\
			${\rm N}_3(0,I_3)$ & ${\rm N}_3(0,{\rm diag}(0.25,1,1))$ & 0.09 & 0.11 & 0.13 & 0.11 & 0.19 & 0.24\\
			${\rm N}_3(0,I_3)$ & ${\rm N}_3(0,\Sigma)$ & 0.09 & 0.16 & 0.23 & 0.18 & 0.47 & 0.70\\
			${\rm N}_3(0,I_3)$ & $t_3$ & 0.08 & 0.10 & 0.09 & 0.12 & 0.16 & 0.14\\
			${\rm N}_3(0,I_3)$ & $t_3/\sqrt{3}$ & 0.11 & 0.22 & 0.28 & 0.24 & 0.48 & 0.59\\
			\hline 
		\end{tabular}
		\caption{Empirical results at the nominal level 5\%, for $k=2$, $d=3$, $n_1=n_2=n$. The upper part is the empirical level and the lower part is the power.} 
		\label{dim3k2}
	\end{table}

	\begin{table}
		\centering
		\begin{tabular}{|c|c|c|c|c|c|c|c|c|}
			\hline
			\multicolumn{3}{|c|}{ } & \multicolumn{3}{c|}{$n=20$ } & \multicolumn{3}{c|}{$n=40$ }\\
			\hline
			$F_1$ & $F_2$ & $r$ & ${\cal E}_1$ & ${\cal E}_{\log}$ & ${\cal E}_0$ & ${\cal E}_1$ & ${\cal E}_{\log}$ & ${\cal E}_0$\\
			\hline
			${\rm N}_3(0,I_3)$ & & 4 & 0.05 & 0.05 & 0.05 & 0.05 & 0.05 & 0.05\\
			${\rm N}_3(0,\Sigma)$ & & 4 & 0.06 & 0.05 & 0.06 & 0.05 & 0.05 & 0.05\\
			$t_3$ & & 4 & 0.05 & 0.05 & 0.06 & 0.05 & 0.05 & 0.06\\ 
			\hline
			${\rm N}_3(0,I_3)$ & ${\rm N}_3((0.5,0.5,0.5),I_3)$ & 3 & 0.58 & 0.54 & 0.50 & 0.92 & 0.89 & 0.85\\
			& & 2 & 0.77 & 0.71 & 0.65 & 0.98 & 0.97 & 0.96\\
			${\rm N}_3(0,I_3)$ & ${\rm N}_3(0,{\rm diag}(0.25,0.25,1))$ & 3 & 0.17 & 0.30 & 0.39 & 0.37 & 0.70 & 0.82\\
			& & 2 & 0.27 & 0.49 & 0.55 & 0.59 & 0.89 & 0.93\\
			${\rm N}_3(0,I_3)$ & ${\rm N}_3(0,\Sigma)$ & 3 & 0.11 & 0.17 & 0.26 & 0.21 & 0.46 & 0.67\\
			& & 2 & 0.14 & 0.24 & 0.35 & 0.26 & 0.63 & 0.82\\
			${\rm N}_3(0,I_3)$ & $t_3/\sqrt{3}$ & 3 & 0.10 & 0.18 & 0.26 & 0.21 & 0.44 & 0.58\\
			& & 2 & 0.15 & 0.30 & 0.38 & 0.34 & 0.65 & 0.75\\
			\hline
		\end{tabular}
		\caption{Empirical results at the nominal level 5\%, for $k=4$, $d=3$, $n_i=n$, $1\leq i \leq 4$, $r$ populations have distribution $F_1$, while the remaining $4-r$ populations have distribution $F_2$. The upper part is the empirical level and the lower part is the power.} \label{dim3k4}
	\end{table}

	\begin{table} 
		\centering
		\resizebox{1\textwidth}{!}{%
			\begin{tabular}{|c|c|c|c|c|c|c|c|}
				\hline
				\multicolumn{2}{|c|}{ } & \multicolumn{3}{c|}{$n=20$ } & \multicolumn{3}{c|}{$n=40$ }\\ \hline $F_1$ & $F_2$ & ${\cal E}_1$ & ${\cal E}_{\log}$ & ${\cal E}_0$ & ${\cal E}_1$ & ${\cal E}_{\log}$ & ${\cal E}_0$\\
				\hline
				$C_1,\ldots,S_5\sim {\rm N}(0,1)$ & $C_1,\ldots,S_5\sim {\rm N}(0,1)$ & 0.05 & 0.05 & 0.05 & 0.05 & 0.05 & 0.05\\
				$C_1,\ldots,S_5\sim t_3$ & $C_1,\ldots,S_5\sim t_3$ & 0.05 & 0.05 & 0.05 & 0.05 & 0.05 & 0.05\\ $C_1,\ldots,S_5\sim \chi^2_1$ & $C_1,\ldots,S_5\sim \chi^2_1$ & 0.05 & 0.05 & 0.05 & 0.05 & 0.05 & 0.05\\
				\hline
				$C_1,\ldots,S_5\sim {\rm N}(0,1)$ & $C_1,\ldots,S_2\sim {\rm N}(0,1)$, $S_3,S_4,S_5\sim {\rm N}(0.5,1)$ & 0.40 & 0.39 & 0.39 & 0.75 & 0.72 & 0.72\\
				$C_1,\ldots,S_5\sim {\rm N}(0,1)$ & $C_1,\ldots,C_5\sim {\rm N}(0,1)$, $S_1,\ldots,S_5\sim 0.5{\rm N}(0,1)$ & 0.16 & 0.38 & 0.39 & 0.37 & 0.79 & 0.81\\
				$C_1,\ldots,S_5\sim {\rm N}(0,1)$ & $C_1,\ldots,S_5\sim t_3/\sqrt{3}$ & 0.07 & 0.16 & 0.16 & 0.13 & 0.36 & 0.38\\
				$C_1,\ldots,S_5\sim {\rm N}(0,1)$ & $C_1,\ldots,S_5\sim (\chi^2_1-1)/\sqrt{2}$ & 0.12 & 0.31 & 0.32 & 0.20 & 0.58 & 0.61\\
				$C_1,\ldots,S_5\sim {\rm N}(0,1)$ & $C_1,\ldots,C_5\sim {\rm N}(0,1)$, $S_1,\ldots,S_5\sim {\rm Logistic}(0,1)$ & 0.31 & 0.64 & 0.64 & 0.75 & 0.97 & 0.97\\
				$C_1,\ldots,S_5\sim {\rm N}(0,1)$ & $C_1,\ldots,S_4\sim {\rm N}(0,1)$, $S_5\sim |{\rm N}(0,1)|$ & 0.38 & 0.37 & 0.37 & 0.71 & 0.72 & 0.72\\
				\hline
		\end{tabular} } \caption{Empirical results at the nominal level 5\%, for $k=2$, functional data, $n_1=n_2=n$. The upper part is the empirical level and the lower part is the power.} \label{fdatak2} \end{table}

	\begin{table}
		\centering
		\resizebox{1\textwidth}{!}{%
			\begin{tabular}{|c|c|c|c|c|c|c|c|c|}
				\hline
				\multicolumn{3}{|c|}{ } & \multicolumn{3}{c|}{$n=20$ } & \multicolumn{3}{c|}{$n=40$ }\\ \hline $F_1$ & $F_2$ & $r$ & ${\cal E}_1$ & ${\cal E}_{\log}$ & ${\cal E}_0$ & ${\cal E}_1$ & ${\cal E}_{\log}$ & ${\cal E}_0$\\
				\hline
				$C_1,\ldots,S_5\sim {\rm N}(0,1)$ & & 4 & 0.05 & 0.05 & 0.05 & 0.05 & 0.05 & 0.05\\
				$C_1,\ldots,S_5\sim t_3$ & & 4 & 0.05 & 0.06 & 0.06 & 0.05 & 0.05 & 0.05\\ $C_1,\ldots,S_5\sim \chi^2_1$ & & 4 & 0.05 & 0.05 & 0.05 & 0.04 & 0.05 & 0.05\\
				\hline
				$C_1,\ldots,S_5\sim {\rm N}(0,1)$ & $C_1,\ldots,S_2\sim {\rm N}(0,1)$, $S_3,S_4,S_5\sim {\rm N}(0.5,1)$ & 3 & 0.36 & 0.34 & 0.34 & 0.73 & 0.70 & 0.70\\
				& & 2 & 0.73 & 0.70 & 0.70 & 0.87 & 0.85 & 0.85\\
				$C_1,\ldots,S_5\sim {\rm N}(0,1)$ & $C_1,\ldots,C_5\sim {\rm N}(0,1)$, $S_1,\ldots,S_5\sim 0.5{\rm N}(0,1)$ & 3 & 0.13 & 0.33 & 0.35 & 0.31 & 0.77 & 0.79\\
				& & 2 & 0.19 & 0.47 & 0.48 & 0.47 & 0.91 & 0.92\\
				$C_1,\ldots,S_5\sim {\rm N}(0,1)$ & $C_1,\ldots,S_5\sim t_3/\sqrt{3}$ & 3 & 0.07 & 0.15 & 0.16 & 0.13 & 0.30 & 0.31\\
				& & 2 & 0.10 & 0.18 & 0.19 & 0.16 & 0.41 & 0.44\\
				$C_1,\ldots,S_5\sim {\rm N}(0,1)$ & $C_1,\ldots,S_5\sim (\chi^2_1-1)/\sqrt{2}$ & 3 & 0.13 & 0.31 & 0.32 & 0.20 & 0.61 & 0.64\\
				& & 2 & 0.13 & 0.35 & 0.38 & 0.25 & 0.75 & 0.78\\
				$C_1,\ldots,S_5\sim {\rm N}(0,1)$ & $C_1,\ldots,C_5\sim {\rm N}(0,1)$, $S_1,\ldots,S_5\sim {\rm Logistic}(0,1)$ & 3 & 0.34 & 0.62 & 0.62 & 0.78 & 0.96 & 0.96\\
				& & 2 & 0.42 & 0.78 & 0.79 & 0.91 & 1.00 & 1.00\\
				\hline
		\end{tabular} }
		\caption{Empirical results at the nominal level 5\%, for $k=4$, functional data, $n_i=n$, $1\leq i \leq 4$, $r$ populations have distribution $F_1$, while the remaining $4-r$ populations have distribution $F_2$. The upper part is the empirical level and the lower part is the power.} 
		\label{fdatak4}
	\end{table}

	\medskip
	
	The simulation results indicate that all three procedures maintain the nominal significance level reasonably well across the scenarios considered. At the same time, they exhibit markedly different power characteristics. The ordinary energy-distance test ${\cal E}_1$ tends to perform best against pure location alternatives. By contrast, the logarithmic procedures ${\cal E}_{\log}$ and ${\cal E}_0$ often show substantially higher power against scale and shape alternatives. The gains are particularly pronounced for alternatives involving changes in covariance structure, departures from normality, or differences in relative dispersion. In many such cases, the improvement becomes more evident as the dimension increases or when functional data are considered. Overall, the results do not suggest uniform superiority of any single procedure. Rather, they indicate that the logarithmic statistics complement existing energy-distance methods by providing increased sensitivity to certain classes of non-location alternatives.
	
	\subsection{Real data set applications}\label{secdata}
	We first consider the well-known Fisher Iris data set, available from the UCI Machine Learning Repository \cite{iris}. The data consist of three samples of size 50 ($k=3$), corresponding to three iris species. Each observation is four-dimensional ($d=4$), comprising sepal length, sepal width, petal length, and petal width measurements. Applying the three tests considered in the previous subsection, all of them reject the null hypothesis of equality of the three population distributions.  In fact, each permutation
	$p$-value was smaller than $0.001$.
	
	We also considered the Wine data set \cite{wine}, available from the UCI Machine Learning Repository. These data are the results of a chemical analysis of wines grown in the same region in Italy but derived from three different cultivars ($k=3$). The analysis determined the quantities of 13 ($d=13$) constituents found in each of the three wine types. The data consist of three samples with sizes $n_1=59$, $n_2=71$ and $n_3=48$. Again, each permutation $p$-value was smaller than $0.001$, thus leading to a strong rejection of the null hypothesis of equality of the population distributions. 
	
	While these examples illustrate the applicability of the proposed procedure in multivariate settings, the main advantage of the present Hilbert-space framework is its direct applicability to functional data. To illustrate this aspect, we consider the Berkeley Growth Data set. This data set contains the heights of 39 boys and 54 girls recorded at 31 irregularly spaced ages between one and eighteen years. The data are available in the {\tt R} package {\tt fda} \cite{fda}. The objective is to test whether the distribution of the growth curves is the same for boys and girls ($k=2$). 
	
	Proceeding as in \cite{ZLX10}, the growth curves were reconstructed using local polynomial smoothing. Each individual curve was smoothed separately, using the common bandwidth $h=0.3674$. Figure~\ref{curves} displays the resulting curves.

	\begin{figure}
		\begin{center}
			\begin{tabular}{cc}
				\includegraphics[width=.45\textwidth]{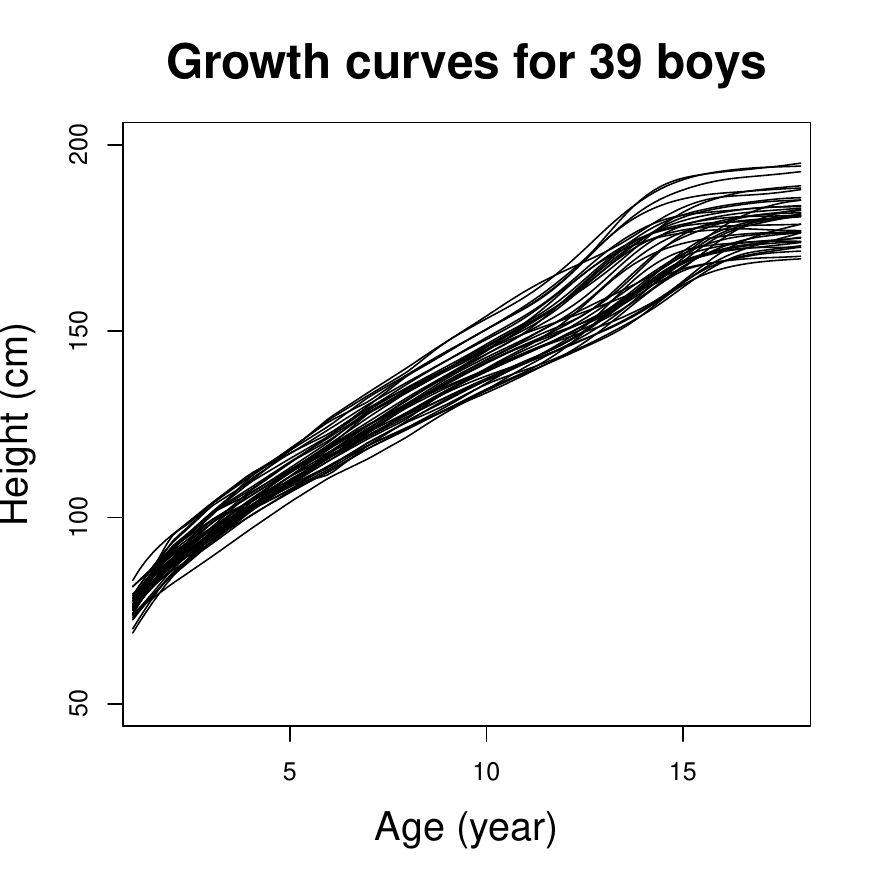} & \includegraphics[width=.45\textwidth]{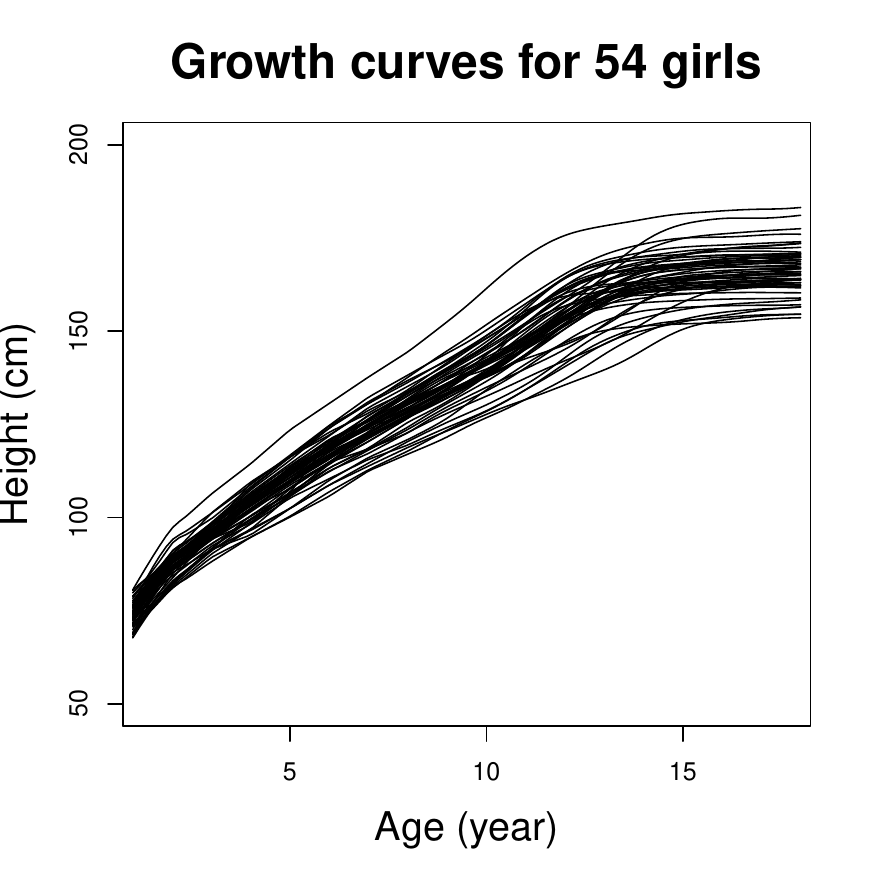}
			\end{tabular}
			\caption{The Berkeley Growth Data.} \label{curves}
		\end{center}
	\end{figure}

	This data set was considered in \cite{ZLX10} for testing equality of mean functions. Following these authors, we consider the four age intervals $[1,4)$, $[4,13)$, $[13,18)$ and $[1,18]$, but now test for equality of distributions rather than equality of means. Table~\ref{growth} reports the resulting permutation $p$-values. For the intervals $[1,4)$, $[13,18)$ and $[1,18]$, all three tests reject the null hypothesis of equality of distributions. These findings are consistent with the results of \cite{ZLX10}, who also reject equality of mean functions on these intervals. By contrast, none of the procedures rejects equality of distributions on the interval $[4,13)$, in agreement with the conclusions of \cite{ZLX10}. Overall, the Berkeley Growth Data provide a natural illustration of the applicability of logarithmic energy distances and logarithmic Gini covariance statistics in a functional-data setting.

	\begin{table}
		\centering
		\begin{tabular}{|l|c|c|c|l|c|c|c|}
			\hline
			interval &   ${\cal E}_1$ & ${\cal E}_{\log}$ & ${\cal E}_0$ &   interval &   ${\cal E}_1$ & ${\cal E}_{\log}$ & ${\cal E}_0$\\ \hline
			$[1,4)$  & 0.011 & 0.025 & 0.027 &   $[13,18)$ & 0.000 & 0.000 & 0.000 \\
			$[4,13)$ & 0.228 & 0.381 & 0.392 &   $[1,18]$  &   0.000 & 0.000 & 0.000\\ \hline
		\end{tabular}
		\caption{Permutation $p$-values obtained from the Berkeley Growth Data set.}
		\label{growth}
	\end{table}
	
	
	\section{Remarks and open problems}\label{secopen}
	
	
	The results of the present paper suggest several directions for future research. A first natural question concerns extensions beyond the Hilbert-space setting. The characterization theorem established in Section~\ref{secloghilbert} relies essentially on Schoenberg's theorem and the resulting representation in terms of Gaussian-kernel maximum mean discrepancies. While the definitions of the logarithmic energy distance and the logarithmic Gini covariance require only a norm, the proof of Theorem~\ref{thmhilbertcharacterization} exploits specific geometric features of Hilbert spaces. It would therefore be of interest to determine to what extent the characterization property and the associated asymptotic theory remain valid in more general Banach spaces. 
	
	Another promising direction concerns high-dimensional asymptotic regimes in which the dimension $d=d_N$ tends to infinity together with the sample size. The logarithmic normalization underlying the proposed statistic may mitigate some of the distance-concentration phenomena affecting ordinary energy statistics in high-dimensional settings. It would therefore be interesting to investigate the behavior of the logarithmic Gini covariance in high-dimension, low-sample-size regimes with $d_N\gg N$, as well as in more classical high-dimensional frameworks. 
	
	A further problem concerns local asymptotic power properties. The simulation results of Section~\ref{secsimul} suggest that the logarithmic statistic may exhibit increased sensitivity to certain classes of shape alternatives and local geometric deviations. A rigorous comparison with ordinary energy-distance procedures and related kernel-based tests therefore appears to be of considerable interest. The logarithmic framework also raises questions concerning adaptive procedures involving several values of the parameter $\alpha$. The statistic studied in the present paper arises as the boundary object corresponding to the limiting regime $\alpha\downarrow 0 $ within the family of $\alpha$-energy and $\alpha$-Gini statistics. It would therefore be natural to investigate data-driven procedures that combine information from several values of $\alpha$, possibly including the logarithmic limit. Such procedures may provide improved sensitivity across different classes of alternatives, ranging from global location shifts to more subtle geometric differences.
	
	Finally, the present work focuses on fixed values of $k$. The original Gini covariance framework of Jim\'enez-Gamero and Sillero-Denamiel~\cite{JGS} was developed in asymptotic regimes where the number of populations tends to infinity. It would therefore be interesting to investigate the behavior of the logarithmic Gini covariance when $k=k_N\to\infty$. Such a study may reveal new phenomena associated with the logarithmic boundary regime and could provide further connections with modern many-sample testing problems.

	\section*{Acknowledgements}
	
	M.D. Jim\'enez-Gamero has been partially
	supported
	by research project PID2022-137818OB-I00 (Ministerio de Ciencia, Innovación y Universidades, Spain).
	This author thanks IMUS-Maria de Maeztu grant CEX2024-001517-M - Apoyo a Unidades de Excelencia María de Maeztu for supporting this research, funded by MICIU/AEI/10.13039/501100011033.

\end{document}